\newif\ifcccg
\ifcccg
  \documentclass{cccg13}
  \usepackage{amsfonts}
  \usepackage{balance} 
\else
  \documentclass[12pt]{article}
  \usepackage{amsthm, fullpage}
  \newtheorem{theorem}{Theorem}
  \newtheorem{lemma}{Lemma}
\fi
\usepackage[pdfborderstyle={/S/U/W 1}]{hyperref}
\usepackage{graphicx,amsmath,enumitem,cite,microtype,caption}
\ifcccg
\else\fi

\setlist{nolistsep}
\newcommand\arr{\mathcal A}
\let\eps\varepsilon

\begin{document}
\title{How To Place a Point to Maximize Angles\thanks{This is an extended version of \cite{c3g}.}}
\author{
	Boris Aronov\thanks{
		Research supported by NSF grants CCF-11-17336 and CCF-12-18791.}\\
	aronov@poly.edu
	\and Mark Yagnatinsky\thanks{
		Research supported by GAANN Grant P200A090157 from the US Department of Education and by NSF grant CCF-11-17336.}\\
  myag@cis.poly.edu
}
\date{\ifcccg\normalsize\fi Polytechnic Institute of NYU, Brooklyn, New York}
\maketitle
\begin{abstract}
  We describe a randomized algorithm that, given a set $P$ of points in the plane, computes the best location to insert a new point $p$, such that the Delaunay triangulation of $P\cup\{p\}$ has the largest possible minimum angle.  The expected running time of our algorithm is at most cubic, improving the roughly quartic time of the best previously known algorithm.  It slows down to slightly super-cubic if we also specify a set of non-crossing segments with endpoints in $P$ and insist that the triangulation respect these segments, i.e., is the constrained Delaunay triangulation of the points and segments.
\end{abstract}

\section{Introduction}
The subject of meshing and specifically constructing ``well behaved'' triangulations has been researched extensively \cite{mesh}.  One of the problems addressed in the literature is that of refining or improving an existing mesh by incremental means.  Motivated by this, Aronov \emph{et~al.} \cite{orig} considered the following problem: ``Given a set of points in the plane, where would you place one additional point, so as to maximize the smallest angle in a good triangulation of the point set?''  Since Delaunay triangulations are known to maximize the smallest angle over all possible triangulations with a given vertex set \cite{Delaunay-good}, the question can be rephrased as: ``Given a point set, where do we place an additional point, so as to maximize the minimum angle in the Delaunay triangulation of the resulting set?''  In the rest of the paper we always picture the new point as lying within the convex hull of the existing points, but the algorithm is essentially the same without this assumption.  (Another variant of the problem mentioned in \cite{orig} involved incrementally improving an existing triangulation by ``tweaking'' the position of an existing interior vertex, one at a time, so that, again, the smallest angle is maximized.)  In \cite{orig}, they also discuss the more challenging question of how to position several points in the best possible coordinated way; we do not address this variant of the problem here.

The previous algorithm \cite{orig} for placing an additional point runs in worst-case $O(n^{4+\eps})$ time, for any~$\eps>0$, with the constant of proportionality depending on~$\eps$.  We propose a randomized algorithm whose expected running time is roughly a factor of $n$ lower.  Somewhat surprisingly, Aronov \emph{et al.} considered and rejected the approach we use in this paper \cite[page~96]{orig}.

The algorithm from \cite{orig} actually handles constrained Delaunay
triangulations, where a set of edges that must be present in the
triangulation is provided as part of the input in addition to a point
set.  This allows one to handle, for example, triangulating a simple
polygon.  We can modify our procedure to deal with constraints in
near-cubic time, slightly slower than the unconstrained case.

We present our algorithm, and show that it runs in cubic time, in the following section.  The analyses of this and the precursor algorithm \cite{orig} are misleading in that they reflect situations unlikely to happen for ``reasonable'' inputs.  We discuss how to measure how realistic an input is, and the resulting behavior of both algorithms on realistic inputs in \autoref{sec:realistic}.  In \autoref{sec:cdt}, we extend our algorithm to handle constrained triangulations in $O(n^2\lambda_{16}(n)\log n)$ time, and we conclude in \autoref{sec:concl}.

\section{The Algorithm}
Our algorithm takes a set $P$ of $n$ points in the plane and computes the best location for a new point~$p$, such that the Delaunay triangulation of $P\cup\{p\}$ has the largest possible minimum angle; for ease of presentation we will assume that the points of $P$ are in \emph{general position}, that is no three points of $P$ lie on a line and no four on a circle.  We start by recalling an argument detailed in \cite{orig} which duplicates the insertion step of the standard incremental Delaunay triangulation algorithm \cite{incremental}.  Let $T$ be the Delaunay triangulation of~$P$.  We begin by computing the arrangement $\arr$ induced by the Delaunay circles of $P$, i.e., of the circumcircles of the triangles of $T$.  (If we are to allow placing $p$ outside the convex hull of $P$, we must deal with $n+1$ points: the input points and the point at infinity.  Thus, $T$ really has $2(n-1)$ triangles: the normal ones and the infinite ones, which are really the lines supporting the convex hull of $P$.  This does not materially affect the algorithm, and for simplicity we ignore this possibility throughout the presentation.)  Although there are only a linear number of Delaunay circles, in the worst case every pair of them intersect, so that $\arr$ has quadratic complexity.  We examine how the Delaunay triangulation~$T_p$ of~$P\cup\{p\}$ differs from $T$.  Let $c$ be the face of~$\arr$ containing~$p$.  Recall that a triangle is present in a Delaunay triangulation if and only if its Delaunay disk is empty of vertices.  Point $p$ \emph{invalidates} some triangles of~$T$ by appearing in the interior of the corresponding disks. After we have inserted~$p$, we no longer have a triangulation; instead we have a star-shaped polygonal hole~$H$ in~$T$ containing~$p$; see \autoref{hole} (left).
\begin{figure*}
\centering
\includegraphics[page=3]{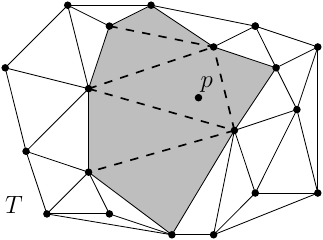}\hfil
\includegraphics[page=5]{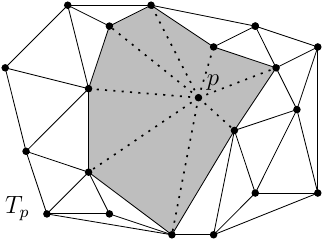}
\caption{\label{hole}The new point $p$ is in the kernel of the shaded star-shaped polygonal hole~$H$.  Removed edges of $T$ are shown dashed (left) and added edges are dotted (right).}
\end{figure*}
Since the insertion of~$p$ only invalidates previously valid triangles, but cannot make an invalid triangle valid (since insertion of $p$ can't turn nonempty disks into empty ones), new edges of~$T_p$ must have~$p$ as an endpoint.  So, connecting $p$ to all vertices of $H$ (\autoref{hole}, right) is the way to complete $T_p$.  This suggests the following algorithm outline:
\begin{enumerate}
\item Compute the Delaunay triangulation $T$.
\item Build the arrangement $\arr$ of Delaunay circles of~$T$.
\item For each of the $O(n^2)$ cells $c \in \arr$:
  \begin{enumerate}
  \item Identify the set of $O(n)$ triangles invalidated by placing $p$ in $c$, the union of which
  	forms the hole~$H$.
  \item Optimize the placement of $p$ in $c$. \label{step}
  \end{enumerate}
\item Return the best placement of $p$ found.
\end{enumerate}
This outline was in fact used in \cite{orig}.  The main contribution of this paper is to use a different approach for \autoref{step}.  Specifically, in \cite{lp}, it was shown that the following is an LP-type problem.\footnote{In \cite{quasi-convex}, this and related problems are presented in a unified framework.}
\begin{quote}
Given a star-shaped polygon $H$, find the point~$p$ in its kernel that maximizes the smallest angle in the triangulation that results by connecting~$p$ to all vertices of~$H$.
\end{quote}
Being an LP-type problem, it can be solved in expected time linear in the number of vertices of $H$, while the approach from \cite{orig}, based on explicitly computing lower envelopes of bivariate functions, takes time roughly quadratic in their number.  However, this LP-type problem is not quite the problem we actually wish to solve, as we need the optimal placement of $p$ \emph{within the current cell~$c$}, which is why this idea was rejected in \cite{orig}.  Fortunately, there is a conceptually simple fix.  In the \emph{region search} stage of our procedure, for each cell~$c$, we run the algorithm from \cite{lp} discarding the result if the returned optimum lies outside~$c$.  A simple argument (see \autoref{one} below) shows that if the solution to the unconstrained problem results in a point not in~$c$, then the optimum within~$c$ must lie on its boundary.  So in a separate \emph{boundary search} step detailed below, we find the best placement of $p$ on any cell boundary.  Combining the results from the two steps we obtain the globally optimal placement for $p$.

\begin{lemma}\label{one}  If the optimal solution to the unconstrained LP-type problem corresponding to cell~$c$ is not in $c$, then the optimal solution for $c$ lies on its boundary.\end{lemma}
\begin{proof} Consider the locus $R(x)$ of points~$p$ such that every angle in the new triangulation of $H$ is at least $x$.  It was shown in \cite{lp} that $R(x)$ is convex; it is easy to see that it varies continuously with~$x$, when non-empty.  Clearly, $R(x) \subset R(y)$ for $y<x$.  As $x$ decreases from its optimum unconstrained value, $R(x)$ will gradually grow from a single point outside~$c$ and eventually intersect~$c$; as it is connected and changes continuously with~$x$, the first intersection must occur along the boundary of~$c$.\end{proof}

It remains to find the best placement for $p$ on each cell boundary.  A cell boundary has two sides, and we process each separately.  First consider an edge of $\arr$.  For a fixed side of a fixed edge $e$, we know which cell of $\arr$ we are in, and thus the hole~$H$.  If $H$ has $k$ vertices, the triangulation has $3k$ angles.  The measure of each of these angles is a univariate function of the position of~$p$ along the edge.  To maximize the smallest of these functions, we find the maximum of their lower envelope by computing the envelope explicitly.  We will show that the graphs of any pair of these functions intersect at most 16 times.  A well-known result from the theory of Davenport-Schinzel sequences immediately implies that the maximum complexity $E(n)$ of the lower envelope is $\lambda_{16}(n)$, where $\lambda_s(n)$ is the maximum length of a DS$(s,n)$ sequence \cite[section 1.2]{ds}.  The maximum length of a DS sequence grows slowly as a function of $n$ when $s$ is constant: it is $o(n\log^*n)$ for any constant $s$.  (It was recently shown in \cite{tight} that a more precise bound is $n\cdot2^{(1+o(1))\alpha(n)^t/t!}$, where $t=\lfloor(s-2)/2\rfloor$.)

\begin{lemma} The complexity of the lower envelope of $n$ angle functions is
$\lambda_{16}(n)$. \end{lemma}
\begin{proof} There are two kinds of angles to consider: angles at the boundary of $H$, and angles at the new point $p$.  We consider first angles at $p$.  Let $p=(x,y)$, and let $q$ and $r$ be two consecutive vertices of $H$; the coordinates of $q$ and $r$ are fixed.  We are interested in the angle $\angle qpr$ at which $p$ sees the segment~$qr$; see \autoref{fig:curve} (left). 
\begin{figure}
  \centering
  \includegraphics{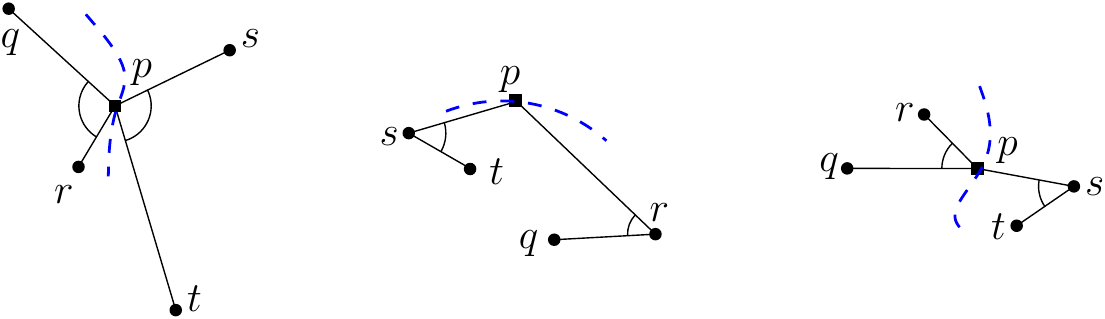}
  \caption{``Artist's impression'' of the curves defined by the three types of angle equality constraints.}
  \label{fig:curve}
\end{figure}
Let $s,t$ be another pair of consecutive vertices.  The angle that $p$ makes with the segment $st$ is $\angle spt$.  Consider the locus of points~$p$ specified by the equation
$\angle qpr = \angle spt$; a point~$p$ satisfying this equation will see $qr$ and $st$ at the same angle; refer to \autoref{fig:curve} (left).  An intersection between this curve and an edge of $\arr$ corresponds precisely to an intersection of the graphs of two angle functions.  Once we prove that there are at most 16 such intersections, we are done.  For convenience, we will equate the cosines of the angles instead of the angles themselves.  Using $|\cdot|$ to denote segment length, the law of cosines gives
$|qr|^2 = |pr|^2 + |pq|^2 - 2|pr||pq|\cos\angle qpr$.  Solving for $\cos\angle qpr$ gives
\[\cos\angle qpr = \frac{|pr|^2 + |pq|^2 - |qr|^2}{2|pr||pq|}.\]
Setting $\cos\angle qpr$ equal to $\cos\angle spt$ produces
\[\frac{|pr|^2 + |pq|^2 - |qr|^2}{|pr||pq|} = \frac{|ps|^2 + |pt|^2 - |st|^2}{|ps||pt|}.\]
After squaring both sides and reshuffling, we obtain
\[(|pr|^2 + |pq|^2 - |qr|^2)^2 |ps|^2 |pt|^2 =(|ps|^2 + |pt|^2 - |st|^2)^2 |pr|^2 |pq|^2.\]
Now each side of the equation is a polynomial in $x$ and $y$ of total degree eight, which means that the locus of points $p$ with $\angle qpr = \angle spt$ is a curve of degree eight.  How many times can such a curve intersect an edge of $\arr$?  An edge is an arc of a circle, which is the zero set of a polynomial of degree two.  According to B\'ezout's  theorem \cite{bez}, the number of proper intersection points is at most the product of the degrees, so there can be at most 16 intersection points.  A similar argument is needed for $\angle qpr = \angle pst$ and also $\angle pqr = \angle pst$ (refer to \autoref{fig:curve} (center and right)), but they also result in polynomial equations of degree at most eight; we omit the entirely analogous calculation.  (In some cases, the degree is only two, but since we are concerned with the worst case, this is little comfort.)  So, the complexity of the envelope is $\lambda_{16}(n)$, and we are done.
\end{proof}

If the worst-case complexity of the lower envelope of $h$~functions from some class is $E(h)$, then we can compute the lower envelope of $n$~functions from that class in $O(E(n)\log n)$ time using a simple divide-and-conquer algorithm \cite[Theorem 6.1]{ds}.  This gives us a running time of $O(\lambda_{16}(n)\log n)$ per arc, which would then make our total running super-cubic if there are a quadratic number of arcs.  However, we are duplicating much work: if we follow a Delaunay circle as it crosses another circle, very little changes when we cross: either one triangle of $T$ ceases to be valid, or else one triangle becomes valid.  (This assumes that we only cross one circle at at time.  At a point of $P$, we may cross many circles at once, so the total change is large, but it is still true that each circle we cross does only one triangle's worth of damage.)  Suppose that a triangle becomes valid when we cross (the other case is symmetric).  Then $H$ loses a boundary vertex, and our triangulation of $H$ loses two old triangles and gains one new one, which means our set of angle functions gains 3 new angles and loses 6 old ones.  The other angle functions remain unchanged.  Instead of restricting the domain of the angle functions to a single arrangement edge, we allow them to be defined wherever the corresponding angle itself exists (so the domain becomes an arc of a Delaunay circle).  On a given arc, there are at most $3n$ functions.  If there are $m$ Delaunay circles, then the boundary of a fixed circle can only have $2(m-1)<2m$ intersections with other circles, and for each of those intersections, at most 6 new functions appear.  The number of circles equals the number of triangles, which is less than $2n$.  Thus for the entire circle, there are at most a linear number of functions ($2n\times2\times6+3n\le27n$).  It is still the case that any pair of function graphs intersect at most 16 times, but because each is not defined over the entire circle, but only a contiguous arc on it, the complexity of the lower envelope can increase slightly, up to $\lambda_{18}(n)$ \cite{ds}.  This is the complexity of an envelope associated with a single circle, and there are $m=O(n)$ circles.  We explicitly compute the $m$ envelopes and find the $m$ associated maxima, $u_1\ldots u_m$.  We also do the region search from the beginning of the section.  The algorithm's final answer is the either the best value that the region search found, or the biggest $u_i$, whichever is larger.  Thus, the running time of the boundary search stage is $O(n\lambda_{18}(n)\log n)$ and the total (expected) running time of our algorithm is dominated by the $O(n^3)$ region search time.  This concludes our description and analysis of the algorithm.

\section{Realistic inputs}\label{sec:realistic}
In the long tradition in computational geometry, exemplified by \cite{realistic}, we would like to be able to analyze our problem in non-worst-case situations.  To this end, we introduce several parameters, besides $n$ that measures the number of input points, that quantify the ``badness'' of the input point set and express the running time of the algorithms in terms of them. 

Consider the arrangement $\arr$ of Delaunay disks of $P$ and let $k$ be its complexity, that is the total number of vertices, edges, and faces; let $d$ be the maximum \emph{depth} of the arrangement, that is the maximum, over all points in the plane, of the number of disks covering the point. In the worst case $k$ is $\Theta(n^2)$ and $d$ is $\Theta(n)$.  In well-behaved point sets, such as those corresponding to uniformly distributed points, $k$ is $\Theta(n)$; one would also
expect $d$ to be near-constant, however, somewhat surprisingly, an unfortunate, but arbitrarily small perturbation of the $\sqrt n\times \sqrt n$ grid can cause~$d$ to be $\Theta(\sqrt n)$, even if we only measure depth within the hull of $P$.  In particular, take the top of the grid, and move the points down slightly so they sit on an upward facing circle.  After a slight perturbation, there are $\sqrt n - 2$ almost identical Delaunay disks which all intersect within the hull of $P$; see \autoref{pt}.
\begin{figure}
  \centering
  \includegraphics{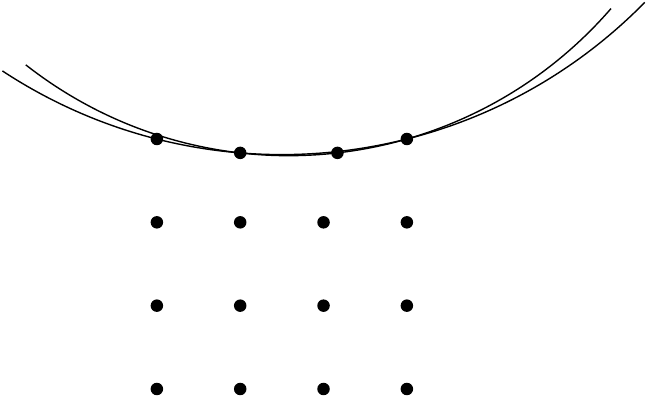}
  \caption{\label{pt}A grid with high depth.}
\end{figure}

We now express the running times in terms of $n$, $k$, and~$d$.  Our algorithm starts by computing the Delaunay triangulation, which can be done in $O(n\log n)$ time.  We then compute the arrangement of circles in $O(k\log n)$ time using a standard sweepline algorithm (better running times are possible using more involved techniques).  Our algorithm and that of \cite{orig} share the first two steps of the outline.  Their analog of the region search runs in time $O(kd^{2+\eps})$, for any positive $\eps$, since for every cell $c \in \arr$, it performs an independent bivariate lower envelope calculation on~$O(d)$ functions, for a total time of $O(kd^{2+\eps} + k\log n)$.  We analyze the region search and the boundary search stages of our proposed algorithm separately.  The region search runs in expected time $O(kd)$, as its bottleneck is solving $k$ LP-type problems of size at most $d$ each.  (Note that this requires that we quickly determine the set of constraints that correspond to a cell.  This is easy to arrange if we traverse the arrangement going from a cell to its immediate neighbor.)

We now turn our attention to the boundary search.  Our analysis here needs stronger general position assumptions than the algorithm itself does.  In particular, we require that if two Delaunay circles intersect in some point not in $P$, no third circle passes through that point.

The running time for one Delaunay circle is affected by how many functions appear on the lower envelope corresponding to that circle.  We earlier derived a bound of $27n$ for the number of functions on a given circle.  We now make this more precise.  Let $f_i$ denote the number of functions along circle $C_i$.  If $C_i$ intersects $x_i$ other circles, and the deepest cell adjacent to $C_i$ has depth $d_i$, then by refining our previous analysis we obtain
$f_i \le 3(d_i+2)+6\cdot 2x_i$.  (If the new point is at depth $d_i$, then the star-shaped hole is composed of $d_i$ triangles and has $d_i+2$ vertices, and the new triangulation will therefore have $d_i+2$ new triangles, and three times as many new angles.  Each time a circle is crossed, six new angles may appear, and each circle is crossed twice.)  Note that since these are Delaunay circles, no disk fully contains another.  Hence, any circle adjacent to a cell of large depth must intersect many other circles.  In particular, $x_i \ge d_i - 1$.  Thus, we have
$f_i \le 3(d_i + 2) + 12x_i \le 3(x_i + 3) + 12x_i = 15x_i+9$, which is $O(x_i)$.

We now show that the sum of $x_i$ over all circles is at most proportional to the arrangement complexity $k$.  Note first that this sum is simply twice the number of pairs of intersecting circles.  Our approach will thus be to show that most pairs of intersecting circles contribute a vertex of degree four to the arrangement $\arr$, that is, a vertex that no third circle passes through.  Indeed, consider a pair of intersecting circles such that both intersection points, call them $q$ and $r$, have degree at least six (in a circle arrangement, all vertices have even degree).  By our stronger general position assumption, both $q$ and $r$ are from the original point set~$P$.  We now have a pair of points with two Delaunay circles passing through it: hence $qr$ must be a Delaunay edge!  But there are only a linear number of such edges, so we are done: all but $O(n)$ pairs of intersecting circles contribute a vertex of degree four to the arrangement, and each such pair can be charged to the vertex.

Finally, let $m$ be the number of circles, $X$ be the number of pairs of intersecting circles, $u$ be the number of vertices of degree four, and $e$ be the number of edges of the Delaunay triangulation.  We now bound the sum of~$x_i$ over all circles:
$\sum_{i=1}^m x_i = 2X \le 2(u + e) = 2u + 2e < 2k + 2e \le 2k + 2(n+m-2) \le
2k + 2(m+2+m-2) = 2k + 4m < 2k + 4k = 6k$,
which is $O(k)$.

Lastly, the total running time of the boundary search stage is at most proportional to
\begin{align*}
\textstyle\sum_{i=1}^m \lambda_{18}(x_i)\log x_i&\le
\textstyle\sum_{i=1}^m \lambda_{18}(x_i)\log m\\&=
\textstyle\log m \cdot \sum_{i=1}^m \lambda_{18}(x_i)\\&\le
\textstyle\log m\cdot\lambda_{18}(\sum_{i=1}^m x_i)\\&\le
\lambda_{18}(6k)\log m,
\end{align*}
which is $O(\lambda_{18}(k)\log n)$.  We thus have the following:
\begin{theorem} Let $P$ be a set of $n$ points in general position, let $k$ be complexity of the arrangement of Delaunay disks induced by $P$, and let $d$ be the maximum depth of this arrangement.  Then the algorithm from the previous section can be implemented to run in time $O(kd + \lambda_{18}(k)\log n)$.
\end{theorem}
Therefore, our algorithm outperforms (in expectation) that of \cite{orig} for most values of $k$ and $d$.  (If $d$ is a constant, their algorithm would take $O(n\log n + k)$ time if implemented carefully, while our boundary search could take time $O(\lambda_{18}(k)\log n)$, which is slightly worse.)

We can slightly refine the above analysis in another direction: recall that we defined $d$ to be the \emph{maximum} depth of the arrangement $\arr$.  If we let $\bar d$ be the \emph{average} depth, over all the cells, the expected running time of the region search can then be bounded by
$O(k\bar d)$, while the running time of the analogous part of algorithm of \cite{orig} is $O(\sum_{c\in\arr} d_c^{2+\eps})$, where $d_c$ is the depth of cell~$c$; the latter quantity is, roughly, $k$ times the average \emph{squared} depth.  The running time of the boundary search is not easily expressed in terms of $\bar d$, but it is less likely to dominate the running time of our algorithm.

\section{Constrained Delaunay triangulations}\label{sec:cdt}
The algorithm from \cite{orig} is actually designed to process
\emph{constrained Delaunay triangulations}: in addition to a point
set, a set of edges is given, which the resulting triangulation must respect.
We now describe how to modify our algorithm to handle this as well.

The first obvious change is that instead of starting with the Delaunay triangulation, we start with the constrained Delaunay triangulation, which can be computed in $O(n\log n)$ time \cite{cdt}.

The second change to the algorithm is that the arrangement $\arr$ must include not only the constrained Delaunay circles, but also the constrained edges.  With this change, knowing the cell containing a new point gives enough information to determine the star-shaped polygonal hole formed by the invalidated triangles \cite{orig}.  However, it is important to actually compute this information quickly.  (Before it sufficed to do $m$ in-circle tests, which tell you which triangles to eliminate.)  However, if we are willing to spend $O(n\log n)$ time, we can simply insert an artificial point in the cell, compute the constrained Delaunay triangulation from scratch, and then compare the resulting triangulation to the triangulation~$T$ to determine which triangles were invalidated.

It remains to handle the boundary search.  We argued above that, as we trace along a circle, crossing another circle does only one triangle's worth of damage.  Unfortunately, when crossing a constrained edge, it is possible that due to changed visibility, many triangles disappear or reappear.  A simple-minded analysis is as follows: at each constrained edge crossing, at most $m < 2n$ triangles appear or disappear.  A circle can only cross $e<3n$ constrained edges.  Thus, there are at most $6n^2$ triangle ``state changes'' along the whole circle boundary.  This analysis is only a constant factor away from being tight: start with $n/3$ nearly co-circular points, and then add $n/3$ constrained segments that each clip the circles; see \autoref{board}.  Thus there is no longer any sense in doing ``whole-circle-at-once'' processing, and the running time of the boundary search is $O(n^2\lambda_{16}(n)\log n)$ by computing envelopes on each arc of the arrangement of Delaunay circles.  This dominates the total running time of the whole algorithm.  (To figure out which angle functions to take the envelope of, we again simply insert an artificial point on the arc in question and re-triangulate from scratch.)  The boundary search can be sped up slightly to $O(n\lambda_{15}(n)\log n)$ by using the algorithm of Hershberger \cite{Hersh}.  So, we can compute the best place to insert a new point, so as to maximize the smallest angle in a constrained Delaunay triangulation in $O(n^2\lambda_{15}(n)\log n)$ time.
\begin{figure}
  \centering
  \includegraphics{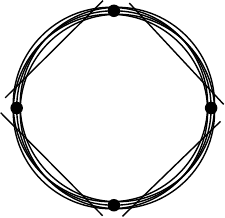}
  \caption{\label{board}Many visibility changes along a circle, for many circles.}
\end{figure}

\section{Conclusions and open problems}\label{sec:concl}
Is there a way to compute which triangles are invalidated in each cell of the arrangement in a way proportional to the depth of that cell in the presence of constraints?  (Our realistic input analysis does not extend to the constrained version of the problem.)

It would be interesting to see if our algorithm can be derandomized using the results of Chazelle and Matou\v sek \cite{derand}; the LP-type problem needs to meet some technical requirements the discussion of which is omitted here.

Can the unconstrained algorithm be sped up by roughly another order of magnitude by observing that there is generally very little difference between LP-type problems corresponding to adjacent cells of $\arr$?  If we were dealing with actual linear programs instead of LP-type problems, we could use dynamic linear programming; see \cite{dynamic} and subsequent improvements.

Is there any hope of generalizing our approach to multiple Steiner points as in \cite{orig}?

\ifcccg\balance\fi
\end{document}